\documentclass[a4paper]{article}
\usepackage[a4paper,top=3cm,bottom=2cm,left=3.5cm,right=3.5cm,marginparwidth=4cm]{geometry}
\usepackage{amsmath,amsfonts}
\usepackage{algorithmic}
\usepackage{algorithm}
\usepackage{array}
\usepackage[caption=false,font=normalsize,labelfont=sf,textfont=sf]{subfig}
\usepackage{textcomp}
\usepackage{stfloats}
\usepackage{url}
\usepackage{verbatim}
\usepackage{graphicx}
\usepackage{cite}
\usepackage{authblk}
\usepackage{mathtools}

\usepackage{mathptmx}
\usepackage{amsthm,amssymb,latexsym}
\usepackage{bbm}

\theoremstyle{plain}
\newtheorem{theorem}{Theorem}
\newtheorem{corollary}[theorem]{Corollary}
\newtheorem{lemma}[theorem]{Lemma}
\newtheorem{proposition}[theorem]{Proposition}

\theoremstyle{definition}

\newtheorem{definition}[theorem]{Definition}

\newtheorem{problem}{Problem} 

\newtheorem{observation}[theorem]{Observation}

\theoremstyle{remark}

\newcommand{\indeg}{\mathrm{indeg}}
\newcommand{\outdeg}{\mathrm{outdeg}}
\newcommand{\head}{\mathrm{head}}
\newcommand{\tail}{\mathrm{tail}}

\newif\ifhighlight
\highlighttrue 
\usepackage{xcolor}
\definecolor{mypink}{rgb}{0.9, 0.0, 0.4}
\definecolor{mygreen}{rgb}{0.0,0.5,0.0}
\definecolor{mypurple}{rgb}{0.6,0.0,0.6}
\definecolor{myfuchsia}{rgb}{0.9,0.0,0.9}
\definecolor{mybrown}{rgb}{0.7,0.3,0.3}
\definecolor{mygray}{rgb}{0.6,0.6,0.6}
\definecolor{myltgray}{rgb}{0.85,0.85,0.85}
\definecolor{stnavy}{HTML}{0D1CC9}
\definecolor{teal}{rgb}{0.0, 0.5, 0.5}
\ifhighlight 
\newcommand{\todo}[1]{\textcolor{red}{\small{[ToDo: #1]}}}
\newcommand{\suggested}[1]{\textcolor{myfuchsia}{\small{[Suggested: #1]}}}

\else 
\newcommand{\todo}[1]{} 
\newcommand{\suggested}[1]{} 
\fi

\begin{document}

\title{Recognizing Level-$k$-Based Phylogenetic Networks is NP-Complete}

\author[1]{Takatora Suzuki}
\affil[1]{Department of Applied Mathematics, Faculty of Science and Engineering, Waseda University, Tokyo 169-8555, Japan}

\maketitle

\begin{abstract}
Phylogenetic networks generalize phylogenetic trees by representing reticulate evolution. 
Tree-based networks and their support trees have been extensively studied, but not all networks are tree-based.
To measure how far such networks are from being tree-based, Suzuki and Hayamizu (2025) formulated the problem of finding the support network with minimum level of a given rooted almost-binary phylogenetic network. 
They conjectured that this problem is NP-hard and provided exponential-time algorithms. 
In this paper, we prove this conjecture by showing that, for every fixed integer $k \geq 1$, it is NP-complete to decide whether the minimum level is at most $k$. 
\end{abstract}

\section{Introduction}\label{sec:intro}

{P}{hylogenetic} networks are graph-theoretical models for representing reticulate evolution of species, such as horizontal gene transfer~\cite{BAPTESTE2013439,szollHosi2015genome} and hybridization~\cite{goulet2017hybridization}. Compared with phylogenetic trees, which describe only branching evolution, they can represent a wider range of evolutionary histories, but this also makes them more difficult to interpret. 
To address this issue, various subclasses of phylogenetic networks have been proposed, each capturing a biologically meaningful scenario while retaining mathematical tractability (see~\cite{kong2022classes} for a comprehensive review). 
Among these, two classes have attracted particular attention.
The first is the class of \emph{level-$k$ networks}~\cite{CHOY200593}, which were introduced to make various computational problems on phylogenetic networks tractable by restricting the local density of reticulations.
Specifically, a network is called level-$k$ if each of its biconnected components contains at most $k$ reticulations. 
Level-$1$ networks, originally defined as galled trees in~\cite{gusfield2003efficient}, have been extensively studied~\cite{gambette2012encodings,marchand_et_al:LIPIcs.WABI.2024.16,BORDEWICH202266}, and level-$k$ networks in general have also attracted considerable attention~\cite{gambette2009structure,IerselEtAl-2009-Constructing}. 
The second is the class of \emph{tree-based networks} (e.g.,~\cite{FrancisSteel-2015-WhichPhylogeneticNetworks,Zhang-2016-TreeBasedPhylogeneticNetworks,anaya2016determining,FischerFrancis-2020-HowTreebasedMy,Hayamizu-2021-StructureTheoremRooted,HayamizuMakino-2023-RankingTopkTrees,pons2019tree-based,tree-based-drawing,Non-binary-tree-based}), which captures the perspective where evolution is tree-like with some reticulate events. 
More precisely, a network is called tree-based if it has at least one \emph{support tree}, defined as a spanning tree sharing the root and leaf set.

Tree-based networks and their support trees have been the subject of active research since their introduction by Francis and Steel~\cite{FrancisSteel-2015-WhichPhylogeneticNetworks}. 
Several studies have shown that whether a given network is tree-based can be decided in linear time~\cite{FrancisSteel-2015-WhichPhylogeneticNetworks,Zhang-2016-TreeBasedPhylogeneticNetworks,Non-binary-tree-based,Hayamizu-2021-StructureTheoremRooted}. 
A particularly influential development is Hayamizu's structure theorem for rooted binary phylogenetic networks (or more generally, for almost-binary networks)~\cite{Hayamizu-2021-StructureTheoremRooted}.
This theorem provides a direct-product characterization of the family of support trees in a given network by canonically decomposing the network into a unique set of subgraphs called maximal zig-zag trails.
This characterization yields optimal algorithms for various computational problems on support trees in a unified manner, such as counting, enumeration, listing, and optimization.
On the other hand, it has also been shown that it is NP-hard to decide whether a given tree-based network has a support tree that is a subdivision of a given tree~\cite{anaya2016determining} (see, e.g.,~\cite{tree-based-drawing} for other hardness result on tree-based networks).

However, not all phylogenetic networks are tree-based. 
Several approaches have been proposed to solve computational problems on such networks~\cite{FischerFrancis-2020-HowTreebasedMy,Hayamizu-2021-StructureTheoremRooted,SuzukiEtAl--BridgingDeviationIndices,fss2018,DavidovEtAl-2021-MaximumCoveringSubtrees,SuzukiEtAl-2025-WhichPhylogeneticNetworks}. 
One natural approach is to consider substructures that generalize support trees and to optimize how ``close to a support tree'' such a substructure is. 
Along this line, Suzuki and Hayamizu~\cite{SuzukiEtAl-2025-WhichPhylogeneticNetworks} introduced the notion of \emph{support networks} for rooted almost-binary phylogenetic networks $N$, defined as spanning subgraphs of $N$ sharing the root and the leaf set. 
They then formulated two optimization problems: for a given rooted almost-binary phylogenetic network~$N$, 
\textsc{Reticulation Minimization} asks for a support network of~$N$ with the minimum number of reticulations, and 
\textsc{Level Minimization} asks for a support network of $N$ with the minimum level. 
The minimum level among all support networks of $N$ is called the \emph{base level} of $N$. 
For the former problem, they provided a linear-time algorithm. 
For the latter problem, they conjectured its NP-hardness and developed two exponential-time algorithms: an exact algorithm and a faster heuristic. 
These algorithms for the two problems are based on their extension of the structure theorem~\cite{Hayamizu-2021-StructureTheoremRooted} from support trees to support networks, which also serves as the theoretical foundation of this paper.

In this paper, we prove that \textsc{Level Minimization} is NP-hard.
We introduce its decision version, which we call \textsc{Level Decision}: given a rooted almost-binary phylogenetic network~$N$ and an integer~$k\geq 0$, decide whether the base level of $N$ is at most $k$. 
We note that the case $k = 0$ coincides with the problem of deciding whether $N$ is tree-based, which is known to be linear-time solvable. 
We prove that \textsc{Level Decision} is NP-complete for every fixed $k \geq 1$; in particular, the NP-hardness of \textsc{Level Minimization} is a direct corollary of this result.

The remainder of this paper is organized as follows. 
Section~\ref{sec:preliminaries} introduces preliminaries, and 
Section~\ref{sec:problem} formulates \textsc{Level Minimization} and 
its decision version \textsc{Level Decision}. 
Section~\ref{sec:related} recalls two key results 
from~\cite{Hayamizu-2021-StructureTheoremRooted} and~\cite{SuzukiEtAl-2025-WhichPhylogeneticNetworks}
and extends them to the multiple-rooted setting, since our reduction uses multiple-rooted networks as gadgets.
Section~\ref{sec:np-hardness} then establishes the NP-completeness of \textsc{Level Decision} for $k=1$ via a polynomial-time reduction from 3SAT, which immediately implies the NP-hardness of \textsc{Level Minimization}. 
Section~\ref{sec:extension} extends this reduction to establish the NP-completeness of \textsc{Level Decision} for every fixed $k\geq 2$. 
Section~\ref{sec:conclusion} concludes with a discussion of open 
problems and directions for future work.

\section{Preliminaries}\label{sec:preliminaries}

\subsection{Graph theoretical terminology}\label{subsec:graph}
The graphs in this paper are finite, simple (i.e.,\ having neither loops nor multiple edges), acyclic directed  graphs, unless otherwise stated. 
For a graph $G$, $V(G)$ and $E(G)$ denote the sets of vertices and edges of $G$, respectively. 
For two graphs $G$ and $H$, $G$ is a \emph{subgraph} of $H$ if both $V(G) \subseteq V(H)$ and $E(G) \subseteq E(H)$ hold, in which case we write $G \subseteq H$. 
Two graphs $G$ and $H$ are \emph{isomorphic}, denoted by $G = H$, if there exists a bijection $\varphi : V(G) \to V(H)$ such that $(u,v) \in E(G)$ if and only if $(\varphi(u), \varphi(v)) \in E(H)$ for all $u, v \in V(G)$. A subgraph $G$ of $H$ is \emph{proper} if $G\neq H$. A subgraph  $G$ of $H$ is a \emph{spanning} subgraph of $H$ if $V(G)=V(H)$.

Given a graph $G$ and a non-empty subset $S \subseteq E(G)$, the edge-set $S$ is said to \emph{induce the subgraph  $G[S]$ of $G$}, that is, the one whose edge-set is $S$ and whose vertex-set is the set of the ends of all edges in $S$. 
For a graph $G$ with $|E(G)| \geq 1$ and a partition  $\{E_1,\dots,E_{d}\}$ of $E(G)$, the collection  $\{G[E_1], \dots, G[E_{d}]\}$ is a  \emph{decomposition} of $G$, and $G$ is said to be \emph{decomposed into} $\{G[E_1], \dots, G[E_{d}]\}$. Here we recall that a partition of a set is a collection of pairwise disjoint non-empty subsets whose union is the entire set.

For an edge $e=(u,v)$ of a graph $G$, $u$ and $v$ 
are denoted by $\tail(e)$ and $\head(e)$, respectively.
For a vertex $v$ of a graph $G$,  the \emph{in-degree of $v$ in $G$}, denoted by $\indeg_G(v)$, is the number of edges $e$ of $G$ with $\head(e) = v$. The \emph{out-degree of $v$ in $G$}, denoted by $\outdeg_G(v)$, is the number of edges $e$ of $G$ with $\tail(e) = v$.

Let $G$ be a graph. \emph{Subdividing} an edge $(u,v)$ of $G$ means replacing it with a directed path from $u$ to $v$ of length at least two. \emph{Smoothing} a vertex $v$ where $\indeg_G(v)=\outdeg_G(v)=1$ means suppressing $v$ from $G$, i.e., the reverse operation of edge subdivision.
A graph $G'$ is called a \emph{subdivision} of $G$ if $G'$ is obtained from $G$ by applying edge subdivisions zero or more times.
For two vertices $u$ and $v$ of a graph~$G$, we call $u$ an \emph{ancestor} of $v$, and $v$ a \emph{descendant} of $u$, if there exists a directed path from $u$ to $v$. 

An undirected graph is \emph{connected} if there is a path between every pair of vertices. 
A directed graph is \emph{connected} if its underlying undirected graph is connected.
For a connected simple undirected graph $G$, a \emph{cut vertex} (resp.\ \emph{cut edge}) of $G$ is a vertex (resp.\ edge) whose removal disconnects $G$. 
A subgraph $H$ of $G$ is called a \emph{block} of $G$ if $H$ is a maximal connected subgraph of $G$ that contains no cut vertex of $H$. 
In this paper, a \emph{block} of a directed graph refers to a block of its underlying undirected graph.
A block is \emph{trivial} if it consists of a single edge or a single vertex; otherwise, it is \emph{non-trivial}.

\subsection{Phylogenetic networks}\label{subsec:phylo net}
Throughout this paper, $X$ represents a non-empty finite set, which can be interpreted as a set of present-day species.
We recall the definitions of rooted almost-binary phylogenetic networks and rooted binary phylogenetic trees from~\cite{Hayamizu-2021-StructureTheoremRooted, SuzukiEtAl-2025-WhichPhylogeneticNetworks}. 

\begin{definition}\label{def:phylo net}
  A \emph{rooted almost-binary phylogenetic network on $X$} is defined to be a finite simple directed acyclic graph $N$ with the following properties (P1)--(P3):
\begin{description}
    \item[(P1)] there exists a unique vertex $\rho$ with in-degree $0$ and out-degree $1$ or $2$ in $N$;
    \item[(P2)] the set of vertices with in-degree $1$ and out-degree $0$ in $N$ is identical to $X$; and
    \item[(P3)] every vertex $v\in V(N)\setminus(\{\rho\}\cup X)$ satisfies $\indeg_N(v) \in \{1,2\}$ and $\outdeg_N(v) \in \{1,2\}$.
\end{description}
When $N$ has properties (P1), (P2) and (P4), $N$ is particularly called a rooted \emph{binary} phylogenetic \emph{tree} on $X$.
\begin{description}
	\item[(P4)] Every vertex $v\in V(N)\setminus(\{\rho\}\cup X)$ satisfies $\indeg_N(v) = 1$ and $\outdeg_N(v) =2$.
\end{description}
\end{definition}
In Definition~\ref{def:phylo net}, we call the vertex $\rho$ \emph{the root} of $N$, each element of $X$ a \emph{leaf} of $N$. We call a non-root and non-leaf vertex $v$ a \emph{reticulation} of $N$ if $\indeg_N(v) = 2$.

Let $N$ be a rooted almost-binary phylogenetic network on~$X$.
If $G$ is a spanning subgraph of $N$ and is a subdivision of some rooted binary phylogenetic tree on~$X$, then $G$ is a \emph{support tree} of $N$. 
If $N$ has a support tree, then $N$ is called \emph{tree-based}.

\begin{definition}\label{def:support net}
  Let $N$ be a rooted almost-binary phylogenetic network on~$X$ and let $G$ be a spanning subgraph of $N$. If $G$ is a rooted almost-binary phylogenetic network on $X$, then $G$ is called a \emph{support network} of $N$. 
A support network $G$ of $N$ is called \emph{minimal} if no support network of $N$ is a proper subgraph of $G$, and is called \emph{minimum} if $G$ has the minimum number of edges among all support networks of $N$.
\end{definition}

Since every support network of $N = (V, E)$ has the same vertex set $V$, we shall identify a support network $G = (V, E')$ of $N$ with its edge-set $E' \subseteq E$ throughout this paper. 
We denote the families of all, minimal, and minimum support networks of $N$ by $\mathcal{A}_N, \mathcal{B}_N$ and $\mathcal{C}_N$, respectively.

For a rooted almost-binary phylogenetic network $N$, 
$\mathrm{level}(N)$ denotes the \emph{level} of $N$, i.e., the maximum number of reticulations contained in a block of $N$. The \emph{base level} of $N$, denoted by $\mathrm{level}^\ast(N)$, is the minimum value of $\mathrm{level}(G)$ over all support networks $G$ of $N$. 
If $k\geq 0$ is the base level of $N$, then $N$ is called \emph{level-$k$-based}. 
We note that $N$ is tree-based if and only if $\mathrm{level}^\ast(N) = 0$ holds since any support tree of $N$ can be viewed as a level-0 support network of $N$.

\section{Problem Definitions and Known Results}
\label{sec:problem}
In this section, we formulate the two problems whose computational complexity is the focus of this paper, \textsc{Level Minimization} (Problem~\ref{prob:level opt}) and its decision version, \textsc{Level Decision} (Problem~\ref{prob:level-k-based decision}).
We then recall a key lemma (Lemma~\ref{lem:BN search correctness}) that will be repeatedly used in our NP-completeness proofs in Sections~\ref{sec:np-hardness} and~\ref{sec:extension}. 

\begin{problem}[\textsc{Level Minimization}, Problem~2 in~\cite{SuzukiEtAl-2025-WhichPhylogeneticNetworks}]
  \label{prob:level opt}
  Given a rooted almost-binary phylogenetic network $N$ on $X$, compute the base level, $\mathrm{level}^\ast(N)$, and find a support network $G\in \mathcal{A}_N$ with $\mathrm{level}(G)=\mathrm{level}^\ast(N)$.
\end{problem}

\begin{figure}[h]
     \centering
     \includegraphics[scale=.6]{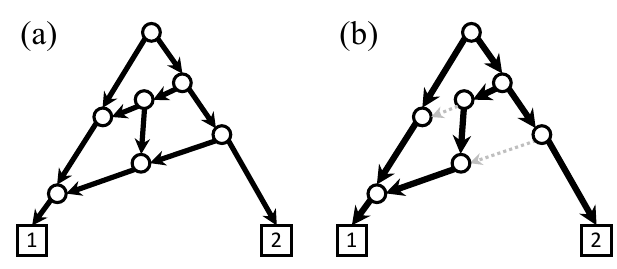}
     \caption{(a) A rooted almost-binary phylogenetic network~$N$ on~$X = \{1, 2\}$, which is not tree-based. (b) A level-1 support network~$G$ of $N$ which minimizes the level among $\mathcal{A}_N$.
     }
     \label{fig:support net}
\end{figure}  

Figure~\ref{fig:support net} illustrates an example of an input-output pair for Problem~\ref{prob:level opt}. For the input network~$N$ shown in Figure~\ref{fig:support net}(a), $N$ has no support tree and has a level-$1$ support network $G$ shown in Figure~\ref{fig:support net}(b). Hence, the output of Problem~\ref{prob:level opt} for this input consists of $\mathrm{level}^\ast(N) = 1$ together with any level-$1$ support network of $N$, such as $G$.

We note that the case $k = 0$ of \textsc{Level Decision}, i.e.\ deciding whether $\mathrm{level}^\ast(N) = 0$, coincides with deciding tree-basedness, which is solvable in linear time for rooted almost-binary networks~\cite{FrancisSteel-2015-WhichPhylogeneticNetworks,Non-binary-tree-based,Zhang-2016-TreeBasedPhylogeneticNetworks,Hayamizu-2021-StructureTheoremRooted}. Hence the NP-hardness of \textsc{Level Minimization} does not follow from any reduction via tree-basedness, in contrast to the unrooted analogue studied by Fischer and Francis~\cite{FischerFrancis-2020-HowTreebasedMy}, whose NP-hardness follows directly from that of deciding tree-basedness for unrooted networks~\cite{francis2018tree}.

Problem~\ref{prob:level-k-based decision} below is the decision version of Problem~\ref{prob:level opt}, and is no harder than Problem~\ref{prob:level opt}, since comparing the output value of Problem~\ref{prob:level opt} with $k$ immediately solves Problem~\ref{prob:level-k-based decision}.

\begin{problem}[\textsc{Level Decision}]
  \label{prob:level-k-based decision}
  Given a rooted almost-binary phylogenetic network $N$ on $X$ and an integer $k\geq 0$, determine whether $\mathrm{level}^\ast(N)\leq k$ or not.
\end{problem}

A key ingredient of our NP-completeness proofs in Sections~\ref{sec:np-hardness} and~\ref{sec:extension} is the following lemma, originally given in~\cite{SuzukiEtAl-2025-WhichPhylogeneticNetworks}. 
It shows that, for both Problems~\ref{prob:level opt} and \ref{prob:level-k-based decision}, the search space can be restricted from $\mathcal{A}_N$ to the (typically much smaller) subfamily $\mathcal{B}_N$ without losing optimality. We will invoke this restriction repeatedly throughout the remainder of the paper.

\begin{lemma}[the first statement of Theorem~19 in~\cite{SuzukiEtAl-2025-WhichPhylogeneticNetworks}]
  \label{lem:BN search correctness}
  For any rooted almost-binary phylogenetic network $N$ on $X$, there exists at least one element $G\in \mathcal{B}_N$ which satisfies $\mathrm{level}(G) = \mathrm{level}^\ast(N)$.
\end{lemma}

As an algorithmic consequence, Lemma~\ref{lem:BN search correctness} guarantees that Algorithm~1 in~\cite{SuzukiEtAl-2025-WhichPhylogeneticNetworks}, which exhaustively searches $\mathcal{B}_N$ in $O(|E(N)| \cdot |\mathcal{B}_N|)$ time, solves Problem~\ref{prob:level opt} exactly. Suzuki and Hayamizu~\cite{SuzukiEtAl-2025-WhichPhylogeneticNetworks} also presented a faster heuristic (Algorithm~2 in~\cite{SuzukiEtAl-2025-WhichPhylogeneticNetworks}) based on an exhaustive search of $\mathcal{C}_N$, which runs in $O(|E(N)| \cdot |\mathcal{C}_N|)$ time but does not always return an optimal solution.
Both algorithms run in exponential time in the worst case.

\section{An Extension to $p$-Rooted Networks:\\ Canonical Decomposition and\\ Characterizations of Support Network Families}%
\label{sec:related}

In this section, we generalize two results from previous works: the maximal zig-zag trail decomposition by Hayamizu~\cite{Hayamizu-2021-StructureTheoremRooted}; and the direct-product characterization of the families of all, minimal, and minimum support networks by Suzuki and Hayamizu~\cite{SuzukiEtAl-2025-WhichPhylogeneticNetworks}; from rooted almost-binary phylogenetic networks to a broader class of networks that may have more than one root.

In what follows, we first define $p$-rooted almost-binary phylogenetic networks in Definition~\ref{def:p-network}, where $p \geq 1$ denotes the number of roots, and restate the first result in this setting (Theorem~\ref{thm:zig-zag trail decomposition}). 
We also recall a characterization of rooted tree-based networks via the decomposition (Proposition~\ref{prop:tree-based}).
We then generalize Definition~\ref{def:support net} to $p$-rooted networks (Definition~\ref{def: p-rooted support network}) and restate the second result (Theorem~\ref{thm:bijection.ABC}).
In the proofs of Theorems~\ref{thm:zig-zag trail decomposition} and~\ref{thm:bijection.ABC}, we describe how the original proofs in~\cite{Hayamizu-2021-StructureTheoremRooted} and~\cite{SuzukiEtAl-2025-WhichPhylogeneticNetworks} can be adapted to the $p$-rooted setting.
This $p$-rooted setting is essential for our NP-completeness proofs in Sections~\ref{sec:np-hardness} and~\ref{sec:extension}, where almost-binary phylogenetic networks with more than one root arise as intermediate objects. A similar multiple-rooted generalization has also been considered in a different context~\cite{huber2022forest}.

\begin{definition}\label{def:p-network}
Let $p \geq 1$ be an integer. A \emph{$p$-rooted almost-binary phylogenetic network on $X$} is defined to be a finite simple directed acyclic graph $N$ with the following properties (P1')--(P3'):
\begin{description}
    \item[(P1')] there exist exactly $p$ vertices with in-degree $0$ and out-degree $1$ or $2$ in $N$;
    \item[(P2')] the set of vertices with in-degree $1$ and out-degree $0$ in $N$ is identical to $X$; and
    \item[(P3')] every other vertex $v$ of $N$ satisfies $\indeg_N(v) \in \{1,2\}$ and $\outdeg_N(v) \in \{1,2\}$.
\end{description}
\end{definition}
In Definition~\ref{def:p-network}, we call each of the $p$ vertices of $N$ described in (P1') a \emph{root} of $N$, the set of all roots of $N$ the \emph{root-set} of $N$, and each element of $X$ a \emph{leaf} of $N$. We call each non-root and non-leaf vertex $v$ of $N$ a \emph{reticulation} of $N$ if $\indeg_N(v) = 2$. 
We note that the case $p =1$ coincides with Definition~\ref{def:phylo net}.
Although any 1-rooted almost-binary phylogenetic network on~$X$ must be connected, we note that a $p$-rooted almost-binary phylogenetic network on~$X$ is not necessarily connected for $p\geq 2$ since different roots may belong to different components (see Figure~\ref{fig:p-rooted}(c) for an example).

Let $p\geq 1$ be an integer and let $N$ be a $p$-rooted almost-binary phylogenetic network on $X$. 
A connected subgraph $Z$ of $N$ with $m\geq 1$ edges is called a \emph{zig-zag trail} if there exists a permutation $(e_1, \dots, e_m)$ of $E(Z)$, where either $\head(e_i) = \head(e_{i+1})$ or $\tail(e_i) = \tail(e_{i+1})$ holds for all $i\in[1, m-1]$. 
In this paper, we often identify $Z$ and the permutation $(e_1, \dots, e_m)$.
A zig-zag trail $Z$ is \emph{maximal} if no zig-zag trail contains $Z$ as a proper subgraph.
A zig-zag trail is represented by an alternating sequence of (not necessarily distinct) vertices and distinct edges such as $(v_0, (v_0, v_1), v_1, (v_2, v_1), \dots, (v_m, v_{m-1}), v_m)$. 
Moreover, if we adopt the notation of each edge $(v_i, v_j)$ as $v_i > v_j$ or $v_j < v_i$, we can more concisely express $Z$ as $v_0 > v_1 < v_2\dots, v_{m-1} < v_m$ or its reverse. 
We call $|E(Z)|$ the \emph{length} of a maximal zig-zag trail~$Z$.

Every maximal zig-zag trail is classified into one of the four types. 
A maximal zig-zag trail $Z$ is called a \emph{crown} if $Z$ has even number $m$ of edges and can be written in the cyclic form $v_0 > v_1 < \dots, v_{m-1} < v_m = v_0$.
A \emph{fence} is a maximal zig-zag trail which is not a crown. 
An \emph{M-fence} is a fence with even number $m$ of edges and of the form $v_0 < v_1 > \dots, v_{m-1} > v_m \neq v_0$. 
A \emph{W-fence} is a fence with even number $m$ of edges and of the form $v_0 > v_1 < \dots, v_{m-1} < v_m \neq v_0$. 
An \emph{N-fence} is a fence with odd number $m$ of edges and of the form $v_0 > v_1 < \dots, v_{m-1} > v_m$. 

As in the approach taken in~\cite{Hayamizu-2021-StructureTheoremRooted, HayamizuMakino-2023-RankingTopkTrees,SuzukiEtAl--BridgingDeviationIndices,SuzukiEtAl-2025-WhichPhylogeneticNetworks}, we often consider a maximal zig-zag trail $Z$ as a sequence $(e_1,\dots,e_{|E(Z)|})$ of edges, ordered according to their appearance in the trail.
For any maximal zig-zag trail $Z = (e_1, \dots,  e_{|E(Z)|})$ in $N$, any subset $S$ of $E(Z)$  is specified by a $0$-$1$ sequence $\langle b_1\; b_2\; \dots \; b_{|E(Z)|}\rangle$, where $b_i = 1$ if $e_i \in S$  and $b_i = 0$ otherwise. For example, given a W-fence $Z=(e_1, e_2, e_3, e_4)$, a subset $\{e_1, e_3, e_4\}$ of $E(Z)$ can be encoded as $\langle 1 0 1 1 \rangle$.

The following Theorem~\ref{thm:zig-zag trail decomposition} was 
originally proved for rooted binary phylogenetic networks by 
Hayamizu~\cite{Hayamizu-2021-StructureTheoremRooted}, and was 
subsequently extended to rooted almost-binary phylogenetic networks 
by Suzuki et al.~\cite{SuzukiEtAl--BridgingDeviationIndices}, as 
noted in Section~6 of~\cite{Hayamizu-2021-StructureTheoremRooted}. We 
further extend it to $p$-rooted almost-binary phylogenetic networks.

\begin{theorem}%
  [\!\!\cite{Hayamizu-2021-StructureTheoremRooted}, $p$-rooted almost-binary version]
  \label{thm:zig-zag trail decomposition}
  Let $p\geq 1$ be an integer and let $N$ be a $p$-rooted almost-binary phylogenetic network on $X$. Then, $N$ is uniquely decomposed into the set $\mathcal{Z} = \{Z_1, \dots, Z_d\}$ of its maximal zig-zag trails. Moreover, this decomposition is computed in $\Theta(|E(N)|)$ time. 
\end{theorem}

\begin{proof}
  The original proof of Theorem~4.2 in~\cite{Hayamizu-2021-StructureTheoremRooted} relies only on the facts that $N$ is almost-binary and finite, neither of which is affected by allowing multiple roots. Similarly, Algorithm~5.1 in~\cite{Hayamizu-2021-StructureTheoremRooted}, which computes the decomposition by visiting each edge exactly once, is unaffected. This completes the proof.
\end{proof}

The maximal zig-zag trail decomposition yields the following characterization of rooted tree-based networks. Equivalent results appear in~\cite{Zhang-2016-TreeBasedPhylogeneticNetworks,Non-binary-tree-based}.

\begin{proposition}[\!\!\cite{Zhang-2016-TreeBasedPhylogeneticNetworks,Non-binary-tree-based,Hayamizu-2021-StructureTheoremRooted}]\label{prop:tree-based}
  Let~$N$ be a rooted almost-binary phylogenetic network on $X$ and let $\mathcal{Z} = \{Z_1, \dots, Z_d\}$ be the maximal zig-zag trail decomposition of $N$. 
  Then, $N$ is tree-based if and only if no element $Z_i\in \mathcal{Z}$ is a W-fence. 
\end{proposition}

\begin{definition}\label{def: p-rooted support network}
  Let $p\geq 1$ be an integer, let $N$ be a $p$-rooted almost-binary phylogenetic network on~$X$ and let $G$ be a spanning subgraph of $N$. If $G$ is a $p$-rooted almost-binary phylogenetic network on $X$, then $G$ is called a \emph{support network} of $N$. 
  A support network of $N$ is called \emph{minimal} or \emph{minimum} in the same sense as in Definition~\ref{def:support net}.
\end{definition}
As in the $1$-rooted setting, we denote the families of all, minimal, 
and minimum support networks of $N$ by $\mathcal{A}_N$, 
$\mathcal{B}_N$, and $\mathcal{C}_N$, respectively. Each of these 
families is non-empty for any $p$-rooted network $N$, since $N$ 
itself is a support network of $N$.
We also note that, since any support network $G$ of $N$ is a spanning subgraph of $N$, no vertex of $G$ can be a root that is not already a root of $N$, and hence $G$ shares the root-set with~$N$.

\begin{figure}[h]
  \centering
  \includegraphics[scale=0.48]{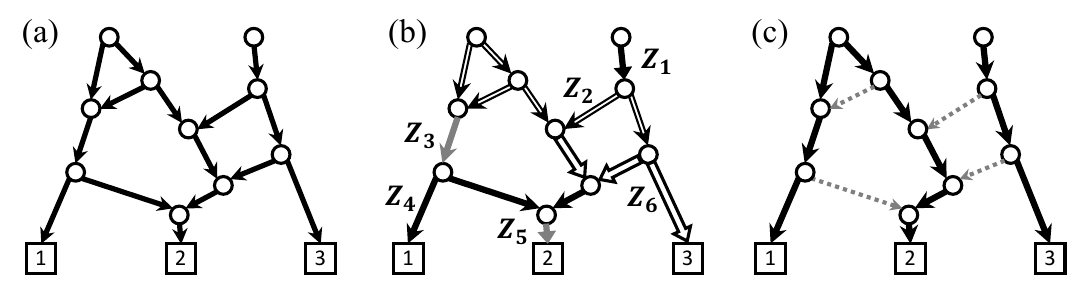}
  \caption{(a) A 2-rooted almost-binary phylogenetic network~$N$ on $X = \{1,2,3\}$. (b) The maximal zig-zag trail decomposition $\mathcal{Z} = \{Z_1, \dots, Z_6\}$ of $N$. All zig-zag trails are distinguished by different arrow styles. (c) A support network of~$N$, which is minimal and minimum.}
  \label{fig:p-rooted}
\end{figure}

  Let $p\geq 1$ be an integer, let $N$ be a $p$-rooted almost-binary phylogenetic network on $X$ and let $Z$ be any subgraph of $N$. A subset $S$ of $E(Z)$ is \emph{$\mathcal A$-admissible} if it satisfies the following conditions:
  \begin{description}
    \item [(C1)] if $(u, v)$ is an edge of $Z$ with $\outdeg_N(u) = 1$ or $\indeg_N(v) = 1$, then $S$ contains $(u,v)$; and
    \item [(C2)] if $e_1$ and $e_2$ are distinct edges of $Z$ with  $\tail(e_1) = \tail(e_2)$ or $\head(e_1) = \head(e_2)$, then $S$ contains at least one of $\{e_1, e_2\}$.
  \end{description}
 In particular, given an $\mathcal A$-admissible subset $S$ of $E(Z)$,  $S$ is \emph{$\mathcal B$-admissible} if it is minimal, i.e.\ no proper subset of $S$ is an $\mathcal A$-admissible subset of $E(Z)$, and  $S$ is \emph{$\mathcal C$-admissible} if it is smallest among all $\mathcal A$-admissible subsets of $E(Z)$.

We now summarize the characterization of $\mathcal{A}_N, \mathcal{B}_N$, and $\mathcal{C}_N$ given in Section~4 of~\cite{SuzukiEtAl-2025-WhichPhylogeneticNetworks} and extend it to $p$-rooted almost-binary phylogenetic networks~$N$ as Theorem~\ref{thm:bijection.ABC} below.

\begin{theorem}%
  [\cite{SuzukiEtAl-2025-WhichPhylogeneticNetworks}, $p$-rooted version]
  \label{thm:bijection.ABC}
  Let $p\geq 1$ be an integer, let $N$ be a $p$-rooted almost-binary phylogenetic network on~$X$, let $\mathcal{Z} = \{Z_1, \dots, Z_d\}$ be the maximal zig-zag trail decomposition of $N$, and let $\mathcal{X} \in \{\mathcal{A}, \mathcal{B}, \mathcal{C}\}$. 
  Then the following hold.
  \begin{enumerate}
    \item There is a one-to-one correspondence between the family $\mathcal{X}_N$ of support networks and the family $\mathcal{S}_{\mathcal{X}}$ of $\mathcal{X}$-admissible subsets of $E(N)$.
    \item The subgraph $N[S]$ of $N$ induced by $S\subseteq E(N)$ is a support network in $\mathcal{X}_N$ if and only if $S \cap E(Z_i)$ is an $\mathcal{X}$-admissible subset of $E(Z_i)$ for each $i\in [1,d]$.
    \item The set $\mathcal{X}_N$ is characterized by $\mathcal{X}_N = \prod_{i=1}^d \mathcal{S}_{\mathcal{X}} (Z_i)$, where $(Z_1, \dots, Z_d)$ is an arbitrary ordering of $\mathcal{Z}$.  
    \item $|\mathcal{X}_N| = \prod_{i=1}^d |\mathcal{S}_{\mathcal{X}} (Z_i)|$ holds. 
  \end{enumerate} 
\end{theorem}

\begin{proof}
We extend the proof in~\cite{SuzukiEtAl-2025-WhichPhylogeneticNetworks} to the $p$-rooted setting by checking that each step is unaffected by allowing multiple roots. For the first statement with $\mathcal{X} = \mathcal{A}$, the proof in Appendix~A.1 of~\cite{SuzukiEtAl-2025-WhichPhylogeneticNetworks} establishes a natural bijection between $\mathcal{A}$-admissible subsets $S$ of $E(N)$ and support networks $N[S]$ of $N$, using the almost-binary condition to identify roots and leaves of $N[S]$ with those of $N$. Since a support network of a $p$-rooted network shares the root-set with $N$ by definition, the argument carries over by replacing ``the root of $N$'' with ``the root-set of~$N$.'' 
For the second statement with $\mathcal{X} = \mathcal{A}$, since $N$ is almost-binary, any two edges referred to in condition~(C2) lie in the same maximal zig-zag trail, so conditions (C1) and (C2) on $S \subseteq E(N)$ decompose into the corresponding conditions on each $S \cap E(Z_i)$ independently; this relies only on the almost-binary condition. The cases $\mathcal{X} \in \{\mathcal{B}, \mathcal{C}\}$ follow similarly, since whether an $\mathcal{A}$-admissible subset of $E(N)$ is minimal or minimum also reduces to the corresponding property on each~$Z_i$. Finally, the third and fourth statements follow immediately from the second. This completes the proof.
\end{proof}

As mentioned in Section~\ref{sec:problem}, we focus on the set $\mathcal{B}_N$.
By Theorem~\ref{thm:bijection.ABC} and the characterization of $\mathcal{S}_\mathcal{B}(Z_i)$ for each type of maximal zig-zag trail $Z_i$, originally given as~(5) in~\cite{SuzukiEtAl-2025-WhichPhylogeneticNetworks} and restated here as~\eqref{eq:sequence for minimal support network}, one can explicitly describe each element of $\mathcal{B}_N$.
\begin{align}
  &\mathcal{S}_{\mathcal{B}}(Z_i) = 
  &\begin{cases}
    \left\{ \langle b_1\cdots b_{|E(Z_i)|}\rangle \;\middle|\; 
      \begin{array}{@{}l@{}}
        b_1 = b_{|E(Z_i)|} = 1, \text{ and} \\
        \text{no $\langle 00\rangle$ or $\langle 111\rangle$ occurs}
      \end{array}
    \right\} \\
    \hfill \text{if $Z_i$ is a fence} \\[2ex]
    \left\{ \langle b_1 \cdots b_{|E(Z_i)|}\rangle \;\middle|\; 
      \begin{array}{@{}l@{}}
        \text{no $\langle 00\rangle$ or $\langle 111 \rangle$ occurs} \\
        \text{in any circular ordering}
      \end{array}
    \right\} \\
    \hfill \text{if $Z_i$ is a crown}
  \end{cases} \label{eq:sequence for minimal support network}
\end{align}

\section{NP-Completeness Proof of Level Decision for $k=1$}\label{sec:np-hardness}
In this section, we prove that Problem~\ref{prob:level-k-based decision} is NP-complete even for $k = 1$ by giving a polynomial-time reduction from \textsc{3SAT}, a classical NP-complete problem ([LO2] in~\cite{garey2002computers}).
Let $\varphi = C_1 \wedge \dots \wedge C_\ell$ be a \emph{3-conjunctive normal form} (3-CNF) formula over a set $U = \{v_1, \dots, v_n\}$ of $n$ Boolean variables, where each clause $C_i = \lambda_{i,1} \vee \lambda_{i,2} \vee \lambda_{i,3}$ consists of three literals $\lambda_{i,j} \in \{v_m, \neg v_m \mid m \in [1,n]\}$.
A \emph{truth assignment} for $U$ is a map $\mathbf{v}: U \to \{\mathsf{T}, \mathsf{F}\}$ assigning a truth value to each variable and $\mathbf{v}$ is said to \emph{satisfy} $\varphi$ if every clause $C_i$ of $\varphi$ evaluates to $\mathsf{T}$ under $\mathbf{v}$.
Given a 3-CNF formula $\varphi$, \textsc{3SAT} asks whether $\varphi$ is \emph{satisfiable}, i.e., whether there exists a truth assignment for $U$ that satisfies $\varphi$.

Our reduction proceeds as follows. First, we introduce \emph{variable gadgets} and \emph{clause gadgets} corresponding to the variables and the clauses of a 3-CNF formula $\varphi$, respectively. We then construct the network $N(\varphi)$ by connecting these gadgets to represent the literal composition of each clause. Finally, we prove the correctness of the reduction by showing that $\varphi$ is satisfiable if and only if $N(\varphi)$ is level-1-based.

\subsection*{Variable gadget}
For each variable $v_m$ of $\varphi$, we associate the rooted almost-binary phylogenetic network $N(v_m)$ on $X_{v_m} = \{x_{m,1}, x_{m,2}\}$, which we call a \emph{variable gadget}, as shown in Figure~\ref{fig:variable gadget}(a).
Intuitively, each $N(v_m)$ is used to encode the truth value of $v_m$ by the location of the non-trivial block in its support network. 
More precisely, we have the following lemma.
\begin{figure}[h]
     \centering
     \includegraphics[scale=.4]{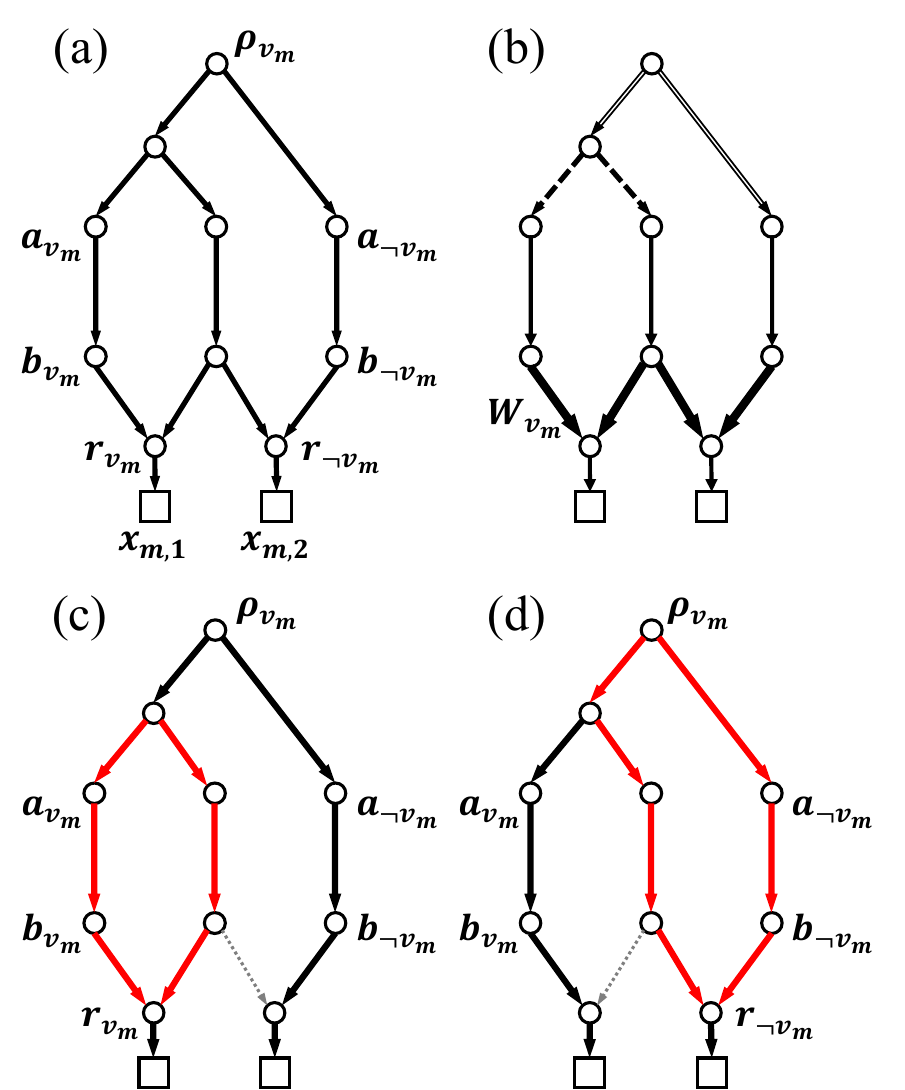}
     \caption{(a) A variable gadget~$N(v_m)$.
     (b) The maximal zig-zag trail decomposition $\mathcal{Z}$ of $N(v_m)$.
     For simplicity, all N-fences of length~$1$ are depicted as thin solid edges, while zig-zag trails of length at least~$2$ are distinguished by different arrow styles.
     (c) A minimal support network $G_\mathsf{T}(v_m)$ of $N(v_m)$.
     (d) A minimal support network $G_\mathsf{F}(v_m)$ of $N(v_m)$.
     In (c) and (d), edges contained in the non-trivial block are shown in red solid line, while the remaining edges are shown in black solid line.}
     \label{fig:variable gadget}
\end{figure}

\begin{lemma}\label{obs:variable gadget}
Let $m\in [1, n]$ and let $N'(v_m)$ be any rooted almost-binary phylogenetic network on~$X_{v_m}$ obtained from $N(v_m)$ shown in Figure~\ref{fig:variable gadget}(a) by subdividing each of the edges $(a_{v_m}, b_{v_m})$ and $(a_{\neg v_m}, b_{\neg v_m})$ zero or more times.
Let $G'_\mathsf{T}(v_m)$ and $G'_\mathsf{F}(v_m)$ be the graphs obtained from $G_\mathsf{T}(v_m)$ and $G_\mathsf{F}(v_m)$, shown in Figure~\ref{fig:variable gadget}(c) and (d), respectively, by applying the same subdivisions.
Then, the following hold.
\begin{enumerate}
  \item $\mathcal{B}_{N'(v_m)} = \{G'_\mathsf{T}(v_m), G'_\mathsf{F}(v_m)\}$ holds.
  \item Each of $G'_\mathsf{T}(v_m)$ and $G'_\mathsf{F}(v_m)$ has one reticulation, $r_{v_m}$ and $r_{\neg v_m}$, respectively.
  \item Every edge on the directed path from $a_{v_m}$ to $b_{v_m}$ is contained in the non-trivial block in $G'_\mathsf{T}(v_m)$, whereas no edge on the directed path from $a_{\neg v_m}$ to $b_{\neg v_m}$ is.
  \item Every edge on the directed path from $a_{\neg v_m}$ to $b_{\neg v_m}$ is contained in the non-trivial block in $G'_\mathsf{F}(v_m)$, whereas no edge on the directed path from $a_{v_m}$ to $b_{v_m}$ is.
\end{enumerate}
\end{lemma}

\begin{proof}
  We first consider the special case $N'(v_m) = N(v_m)$, i.e., no subdivisions are applied.
As shown in Figure~\ref{fig:variable gadget}(b), the maximal zig-zag trail decomposition $\mathcal{Z}$ of $N(v_m)$ consists of one W-fence of length 4 denoted by $W_{v_m}$, two M-fences of length 2, and five N-fences of length 1. By~(1), $W_{v_m}$ has two $\mathcal{B}$-admissible edge-sets $\langle 1011 \rangle$ and $\langle 1101 \rangle$, while each of the other maximal zig-zag trails $Z \in \mathcal{Z}$ has only one $\mathcal{B}$-admissible edge-set, namely $E(Z)$. Therefore, by the third statement of Theorem~\ref{thm:bijection.ABC}, $N(v_m)$ has exactly two minimal support networks $G_\mathsf{T}(v_m)$ and $G_\mathsf{F}(v_m)$, corresponding to the edge subsets $\langle 1011 \rangle$ and $\langle 1101 \rangle$ of $W_{v_m}$, respectively. This proves the first statement. The other statements can then be verified directly by inspecting the non-trivial blocks highlighted in red in Figure~\ref{fig:variable gadget}(c) and~(d). 

The general case follows since the prescribed subdivisions only replace each of the N-fences $a_{v_m} > b_{v_m}$ and $a_{\neg v_m} > b_{\neg v_m}$ in $\mathcal{Z}$ with multiple N-fences of length 1, and thus do not affect the argument above.
This completes the proof.
\end{proof}

\subsection*{Clause gadget}
For each clause $C_i$ of $\varphi$, we associate the 6-rooted almost-binary phylogenetic network $N(C_i)$ on $X_{C_i} = \{y_{i,j}\mid j\in [1,7]\}$, which we call a \emph{clause gadget}, as shown in Figure~\ref{fig:clause gadget}(a). 

\begin{figure}[h]
     \centering
     \includegraphics[scale=0.38]{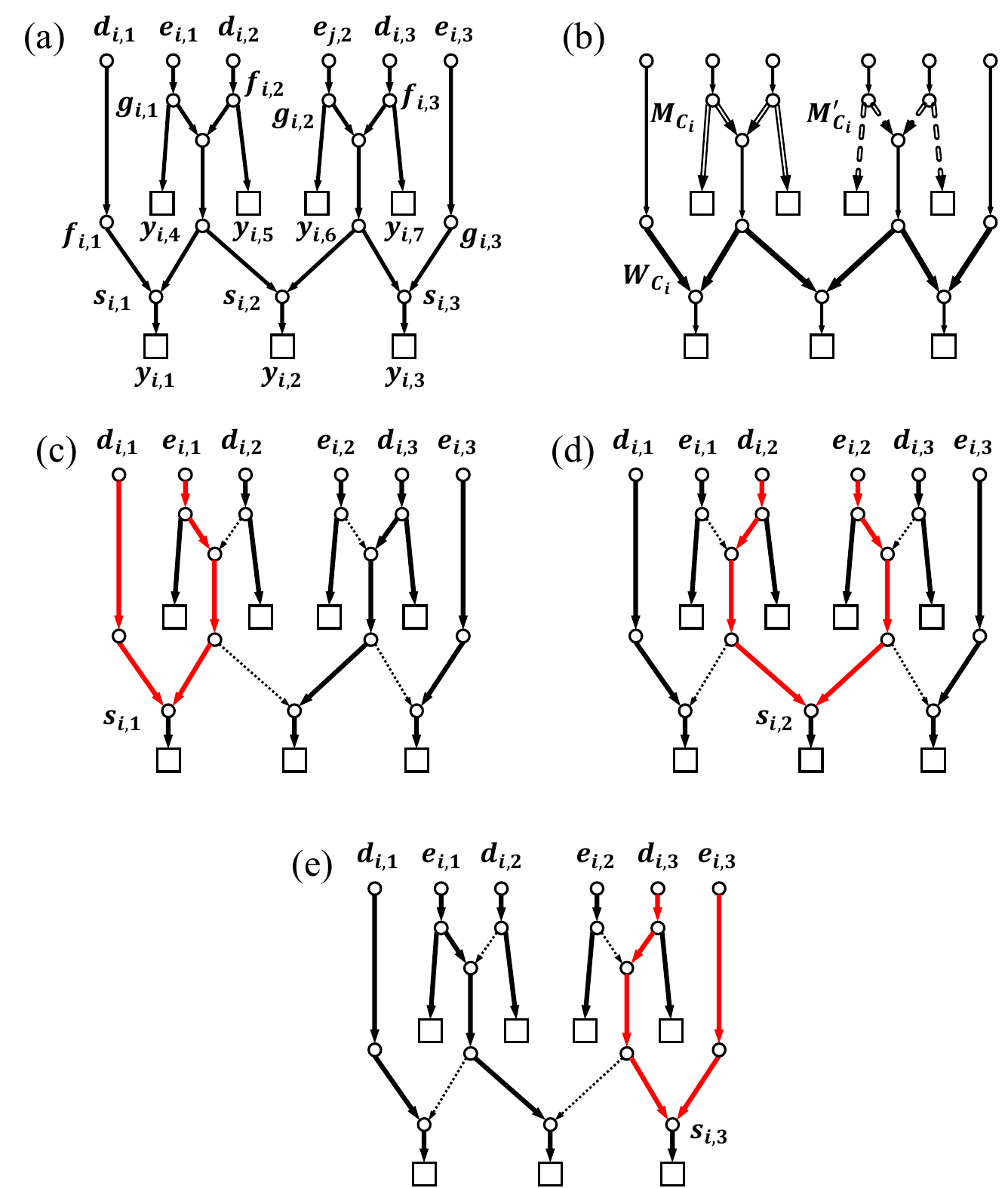}
     \caption{(a) A clause gadget $N(C_i)$.
     (b) The maximal zig-zag trail decomposition $\mathcal{Z}^\prime$ of $N(C_i)$. For simplicity, all N-fences of length~$1$ are depicted as thin solid edges and the other zig-zag trails are distinguished by different arrow styles.
     (c) A minimal support network $G_1(C_i)$ of $N(C_i)$.
     (d) A minimal support network $G_2(C_i)$ of $N(C_i)$.
     (e) A minimal support network $G_3(C_i)$ of $N(C_i)$.
     In (c), (d), and (e), edges that may lie on a directed path to a reticulation are shown in red solid line, while the remaining edges are shown in black solid line.}
     \label{fig:clause gadget}
\end{figure}

\begin{lemma}\label{obs:clause all}
Let $i\in [1, \ell]$ and let $N(C_i)$ be a clause gadget. Then, every support network $G \in \mathcal{B}_{N(C_i)}$ has exactly one reticulation and the reticulation of $G$ has exactly two of the roots of $G$ as ancestors.
\end{lemma}
\begin{proof}
As shown in Figure~\ref{fig:clause gadget}(b), the maximal zig-zag trail decomposition $\mathcal{Z}^\prime$ of $N(C_i)$ consists of two M-fences of length~$4$, denoted by $M_{C_i}$ and $M^\prime_{C_i}$, one W-fence of length~$6$, denoted by $W_{C_i}$, and 11 N-fences of length~$1$. By~\eqref{eq:sequence for minimal support network}, both $M_{C_i}$ and $M^\prime_{C_i}$ have two $\mathcal{B}$-admissible edge-sets $\langle 1011\rangle$ and $\langle 1101\rangle$, neither of which produces a reticulation; $W_{C_i}$ has three $\mathcal{B}$-admissible edge-sets $\langle 101011\rangle$, $\langle 101101\rangle$, and $\langle 110101\rangle$, each of which produces exactly one reticulation. Each of the other maximal zig-zag trails $Z\in \mathcal{Z}^\prime$ has only one $\mathcal{B}$-admissible edge-set, namely $E(Z)$, which produces no reticulation. Therefore, by the third statement of Theorem~\ref{thm:bijection.ABC}, each $G \in \mathcal{B}_{N(C_i)}$ contains exactly one reticulation, which arises in $W_{C_i}$. Let $r$ denote this unique reticulation of $G$. Tracing the two incoming edges at $r$ backward in $G$ yields two internally vertex-disjoint directed paths that terminate at two distinct roots of $G$, proving that exactly two of the six roots of $G$ are ancestors of $r$.
\end{proof}

By the fourth statement of Theorem~\ref{thm:bijection.ABC}, $N(C_i)$ has exactly $2^2\cdot 3 = 12$ minimal support networks, all of which need to be taken into account when $\varphi$ is unsatisfiable. When $\varphi$ is satisfiable, however, three specific ones among them play a pivotal role.
We define the three minimal support networks of $N(C_i)$ illustrated in Figure~\ref{fig:clause gadget}(c), (d), and (e) as $G_1(C_i)$, $G_2(C_i)$, and $G_3(C_i)$, respectively.
In each of these figures, the directed paths highlighted in red indicate the correspondence between the pair of roots and the reticulation of $N(C_i)$.
This correspondence reflects the fact that $C_i$ consists of three literals and that satisfying $C_i$ requires at least one of them to be true, as stated in the following observation.

\begin{observation}\label{obs:clause specific three}
  Let $i\in [1,\ell]$ and let $N(C_i)$ be a clause gadget.
Then, for each $j \in [1,3]$, only the two roots $d_{i,j}$ and $e_{i,j}$ have the reticulation $s_{i,j}$ as a descendant in $G_j(C_i)$, while the remaining four roots do not have any reticulation as a descendant in $G_j(C_i)$. 

\end{observation}

\begin{figure}[t]
    \centering
    \includegraphics[width=0.9\textwidth]{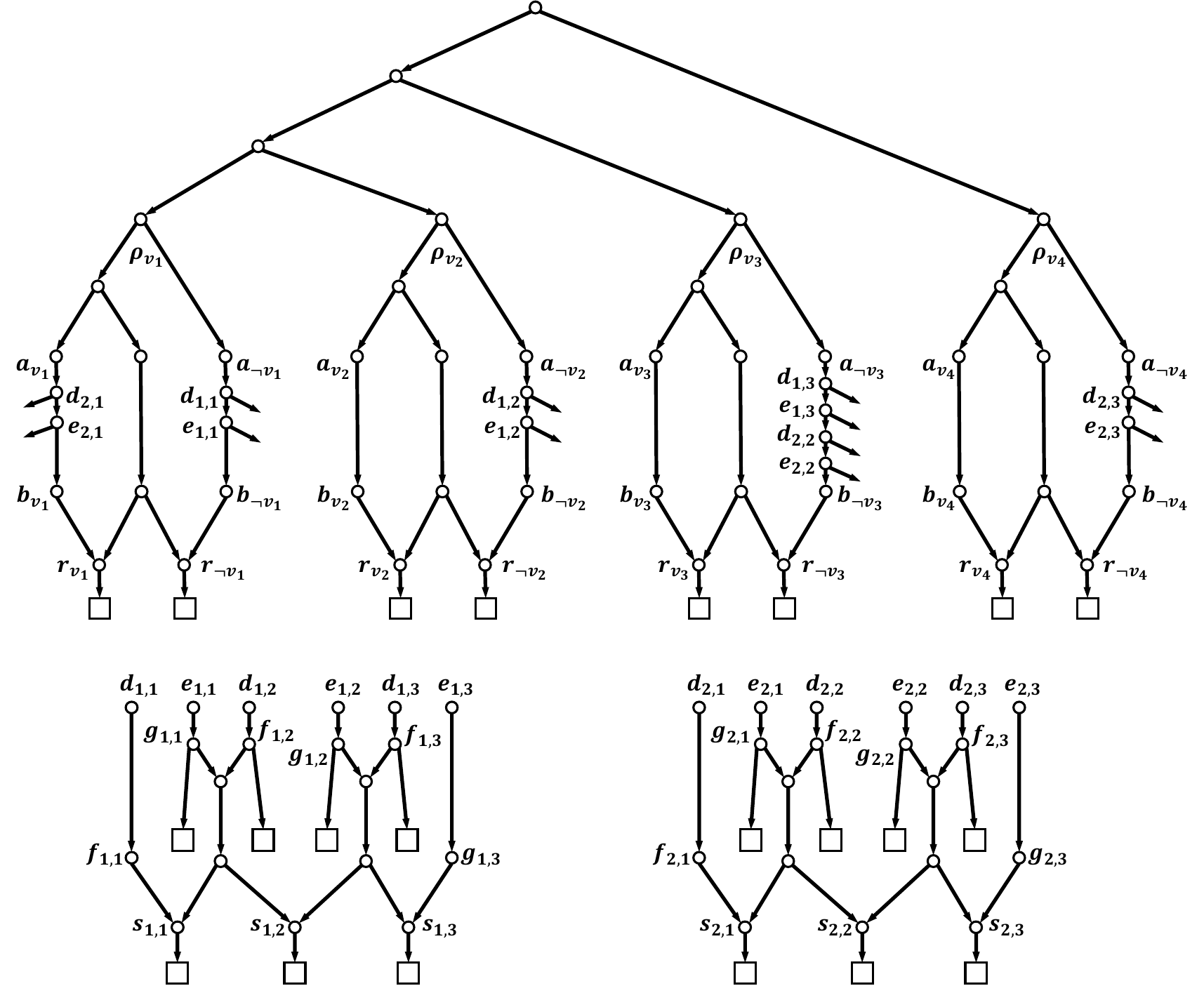}
    \caption{An illustration of $N(\varphi)$
    constructed from the 3-CNF formula $\varphi = (v_1 \vee v_2 \vee v_3)\wedge(\neg v_1 \vee v_3 \vee v_4)$.
    For clarity of presentation, the variable gadgets and the clause gadgets are drawn separately, although they are connected by identifying each pair of vertices labeled with the same label.}
    \label{fig:whole network}
\end{figure}

\subsection*{The construction of $N(\varphi)$}\label{apd:construction}
For a given 3-CNF formula $\varphi$ with $n$ variables, we construct $N(\varphi)$
as follows.
\begin{enumerate}
  \item Pick a rooted binary phylogenetic tree $T$ on $X = \{1, \dots, n\}$ arbitrarily.
  \item Starting from $T$, attach the root $\rho_{v_m}$ of $N(v_m)$ to the leaf $m$ of $T$ for each $m \in [1,n]$.
  \item Subdivide the unique incoming edge at $b_{\neg v_m}$ (if $\lambda_{i,j} = v_m$) or $b_{v_m}$ (if $\lambda_{i,j} = \neg v_m$) by two vertices, and identify each of them with $d_{i,j}$ and $e_{i,j}$ of $N(C_i)$ for each literal $\lambda_{i,j}$ of $\varphi$. Repeating this procedure for all $i \in [1,\ell]$ and $j \in [1,3]$ yields the desired graph $N(\varphi)$.
\end{enumerate}

We remark that, for any 3-CNF formula $\varphi$, the graph $N(\varphi)$ constructed above is a rooted almost-binary phylogenetic network on $X_\varphi = \bigcup_{m\in[1,n]}X_{v_m}\cup \bigcup_{j\in[1,\ell]}X_{C_j}$.
Indeed, when the root of each variable gadget $N(v_m)$ is attached to a leaf of $T$, or when one of the six roots of each clause gadget $N(C_i)$ is attached to a vertex created by subdividing an edge of $N(v_m)$, the attachment transforms the corresponding vertex into one with in-degree~1 and out-degree~2.
The degrees of all other vertices remain unchanged throughout the construction.

As an example, Figure~\ref{fig:whole network} illustrates the rooted almost-binary phylogenetic  network $N(\varphi)$ constructed from $\varphi = (v_1 \vee v_2 \vee v_3)\wedge(\neg v_1 \vee v_3 \vee v_4)$.
Here, the literal $\lambda_{1,1} = v_1$ is represented by connecting the two roots $d_{1,1}$ and $e_{1,1}$ of $N(C_1)$ to the two vertices obtained by subdividing the edge $(a_{\neg v_1}, b_{\neg v_1})\in E(N(v_1))$.

\subsection*{Correctness of the reduction}
For each $m\in [1, n]$, let $N^\prime(v_m)$ be the subdivision of $N(v_m)$ obtained by applying only the subdivisions of $(a_{v_m}, b_{v_m})$ and $(a_{\neg v_m}, b_{\neg v_m})$ as prescribed in step 3 (i.e., without connecting the clause gadgets to the variable gadgets).
Similarly, let $G_\mathsf{T}'(v_m)$ and $G_\mathsf{F}'(v_m)$ be the networks obtained from $G_\mathsf{T}(v_m)$ and $G_\mathsf{F}(v_m)$, respectively, by applying such subdivisions.
Then, the network $N(\varphi)$ can be decomposed into the tree $T$, the variable gadgets $N'(v_m)$ for $m \in [1,n]$, and the clause gadgets $N(C_i)$ for $i \in [1,\ell]$, since the construction of $N(\varphi)$ can be viewed as merely connecting these components.
This decomposition also allows us to characterize the family of minimal support networks of $N(\varphi)$ independently for each component, as stated in the following lemma.

\begin{lemma}\label{lem:construction}
Let $\varphi$ be a 3-CNF formula.
Then, the family $\mathcal{B}_{N(\varphi)}$ of minimal support networks of $N(\varphi)$
is characterized by
\begin{align}
    \mathcal{B}_{N(\varphi)}
    = \{E(T)\} \times \prod_{m=1}^n \mathcal{B}_{N^\prime(v_m)}
      \times \prod_{i=1}^\ell \mathcal{B}_{N(C_i)}.
    \label{eq:lemma 3.5}
\end{align}
\end{lemma}

\begin{proof}
Let $\mathcal{Z}$ be the maximal zig-zag trail decomposition of $N(\varphi)$.
By Theorem~\ref{thm:bijection.ABC}, we have
$\mathcal{B}_{N(\varphi)} = \prod_{Z_i \in \mathcal{Z}} \mathcal{S}_{\mathcal{B}}(Z_i)$.
Moreover, \eqref{eq:sequence for minimal support network} implies that if the length of a maximal zig-zag trail~$Z$ is at most~$2$, then only $E(Z)$ is the $\mathcal{B}$-admissible edge-set of $Z$.
Therefore, it suffices to show that no maximal zig-zag trail of $N(\varphi)$ of length at least~$3$ is newly created or destroyed when $N(\varphi)$ is constructed by connecting the variable gadgets and the clause gadgets to any $T$.

In step~1 of the construction, since any rooted binary phylogenetic tree $T$ on $X = \{1,\dots, n\}$ consists only of $n-1$ M-fences of length~$2$, it contains no maximal zig-zag trail of length at least~$3$.
In step~2, attaching the roots of the variable gadgets to the leaves of $T$ does not change the maximal zig-zag trail decomposition, since, for each $m$, the two outgoing edges from $\rho_{v_m}$ belong to the same maximal zig-zag trail, whereas the incoming edge to $\rho_{v_m}$ belongs to a different maximal zig-zag trail.
In step~3, suppose that a positive literal $v_m$ appears $q$ times in $\varphi$.
Then, the edge $(a_{\neg v_m}, b_{\neg v_m})$ of $N(v_m)$ is subdivided $2q$ times.
The edge $(a_{\neg v_m}, b_{\neg v_m})$ and all edges of the clause gadgets attached
by these subdivisions form only N-fences of length~$1$.
Accordingly, this step transforms these maximal zig-zag trails into
one N-fence of length~$1$ and $2q$ M-fences of length~$2$.
The analogous argument holds when a negative literal $\neg v_m$ appears $q^\prime$ times in $\varphi$. In this case, subdividing the edge $(a_{v_m}, b_{v_m})$ and attaching the corresponding clause gadgets transforms the N-fences of length 1 into one N-fence of length~$1$ and $2q^\prime$ M-fences of length~$2$.
Hence, step~3 neither creates nor destroys any maximal zig-zag trail of length at
least~$3$.
This completes the proof.
\end{proof}

To illustrate the transformation of maximal zig-zag trails in step 3, consider the formula $\varphi = (v_1 \vee v_2 \vee v_3)\wedge(\neg v_1 \vee v_3 \vee v_4)$.
Here, the variable $v_1$ appears once as a positive literal, so the N-fence $a_{\neg v_1}> b_{\neg v_1}$ in $N(v_1)$ and the N-fences $d_{1,1} > f_{1,1}$ and $e_{1,1} > g_{1,1}$ in $N(C_1)$ are transformed into one N-fence $a_{\neg v_1} > d_{1,1}$ and two M-fences $e_{1,1} < d_{1,1} > f_{1,1}$ and $b_{\neg v_1} < e_{1,1} > g_{1,1}$, as shown in Figure~\ref{fig:whole network}.

Inspecting the maximal zig-zag trail decomposition $\mathcal{Z}$ of $N(\varphi)$ considered in the proof of Lemma~\ref{lem:construction}, we see as a byproduct that $\mathcal{Z}$ always contains at least one W-fence arising from a variable gadget, regardless of whether $\varphi$ is satisfiable.
Together with Proposition~\ref{prop:tree-based}, we therefore obtain the following.

\begin{proposition}\label{prop:Nphi is not TB}
  For every 3-CNF formula $\varphi$, we have $\mathrm{level}^{\ast}(N(\varphi)) \geq 1$.
\end{proposition}

We are ready to prove that $\varphi$ is satisfiable if and only if the network $N(\varphi)$ is level-$1$-based, for any 3-CNF formula $\varphi$.

\begin{figure}[t]
    \centering
    \includegraphics[width=0.9\textwidth]{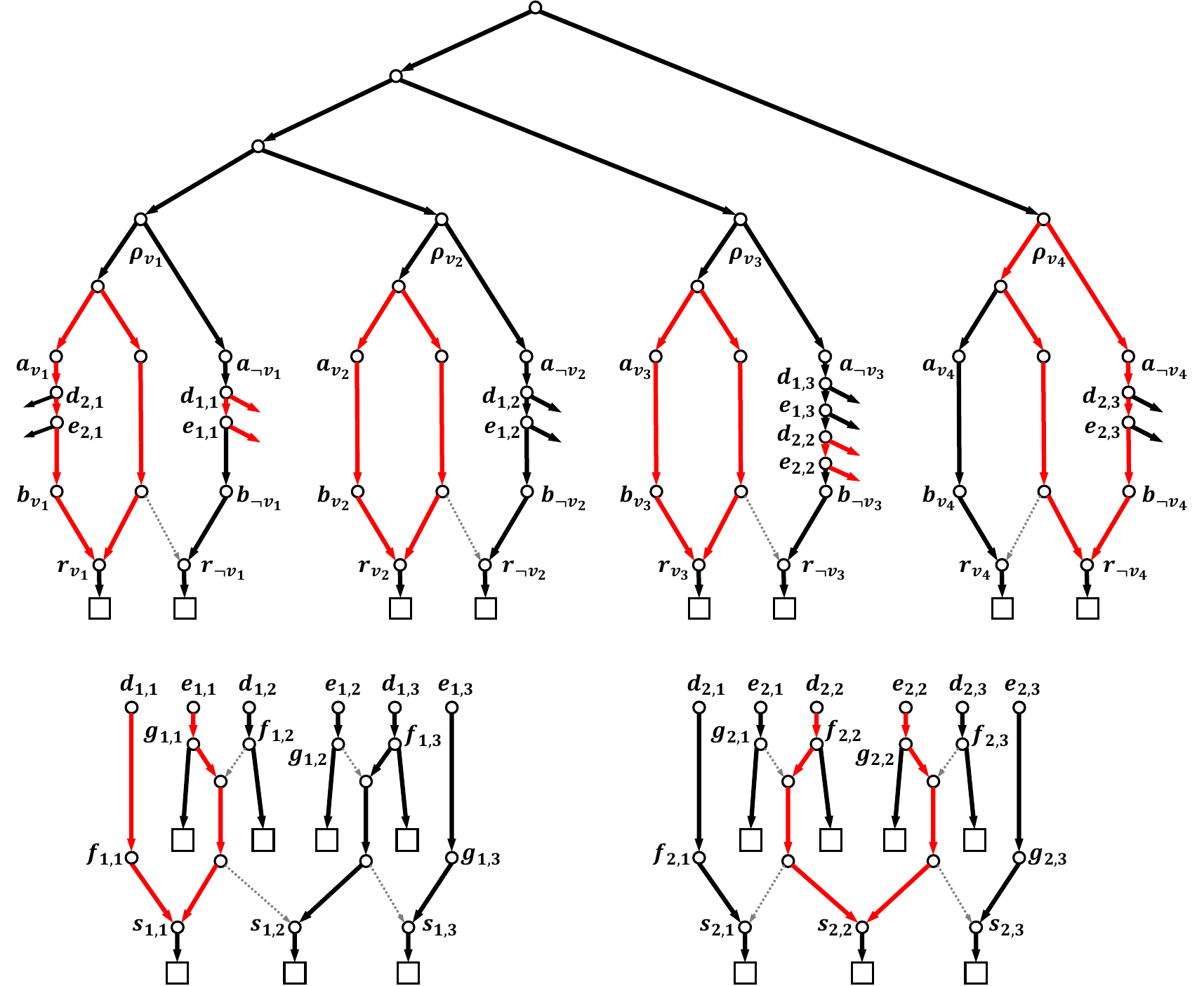}
    \caption{The level-$1$ support network $G$ of $N(\varphi)$ (shown in Figure~\ref{fig:whole network}) constructed following the procedure described in the proof of Theorem~\ref{lem:satisfiable} for $\varphi = (v_1 \vee v_2 \vee v_3)\wedge(\neg v_1 \vee v_3 \vee v_4)$.
    For clarity of presentation, the variable gadgets and the clause gadgets are drawn separately as in Figure~\ref{fig:whole network}. All the edges of $G$ are colored in the same manner as in Figure~\ref{fig:variable gadget}(c) and (d).
    }
    \label{fig:whole network colored}
\end{figure}

\begin{theorem}
\label{lem:satisfiable}
If a 3-CNF formula $\varphi$ is satisfiable, then $N(\varphi)$ has at least one level-$1$ support network.
\end{theorem}

\begin{proof}
Assume that $\varphi$ is satisfiable, and fix a truth assignment $\mathbf{v}$ that satisfies $\varphi$. 
For each $m\in [1, n]$, we write $v_m^\ast \coloneqq \mathbf{v}(v_m)$.
Since $\mathbf{v}$ satisfies every clause $C_i$, for each $i \in [1,\ell]$ there exists an index $j \in [1,3]$ such that the literal $\lambda_{i,j}$ evaluates to $\mathsf{T}$ under $\mathbf{v}$; we fix one such index and denote it by $j_i$.
We construct a support network $G$ of $N(\varphi)$ in the same manner as in the construction of $N(\varphi)$, by combining the tree $T$, the minimal support network $G'_{v_m^\ast}(v_m) \in \mathcal{B}_{N'(v_m)}$ for each $m \in [1,n]$, and the minimal support network $G_{j_i}(C_i) \in \mathcal{B}_{N(C_i)}$ for each $i \in [1,\ell]$.
Lemma~\ref{lem:construction} guarantees that $G\in \mathcal{B}_{N(\varphi)}$ holds. 

Then, we prove that $G$ is level-$1$, which is equivalent to showing that no two reticulations in $G$ belong to the same block. 
Note that the set of reticulations in $G$ consists of exactly one of $\{r_{v_m}, r_{\neg v_m}\}$ for each $m\in[1, n]$, and $s_{i,{j_i}}$ for each $i \in [1,\ell]$.
Let $i \in [1,\ell]$.
Observation~\ref{obs:clause specific three} implies that only $d_{i,{j_i}}$ and $e_{i,{j_i}}$ have a reticulation $s_{i,{j_i}}$ as a descendant in $G$  among the six roots of $C_i$.
Then, suppose that $\lambda_{i,{j_i}}$ is a positive literal $v_m$.
By the construction of $N(\varphi)$, the vertices $d_{i,{j_i}}$ and $e_{i,{j_i}}$ are located on the directed path from $a_{\neg v_m}$ to $b_{\neg v_m}$.
By Lemma~\ref{obs:variable gadget}, 
the directed path from 
$a_{\neg v_m}$ to $b_{\neg v_m}$ is not contained in the non-trivial block of $G'_\mathsf{T}(v_m)$.
Hence, $s_{i,{j_i}}$, which hangs from this path, does not share a block with any other reticulation in~$G$.
If instead $\lambda_{i,{j_i}}$ is a negative literal $\neg v_m$, then $d_{i,{j_i}}$ and $e_{i,{j_i}}$ lie on the directed path from $a_{v_m}$ to $b_{v_m}$, which by Lemma~\ref{obs:variable gadget} is not contained in the non-trivial block of $G'_\mathsf{F}(v_m)$, and the same conclusion follows.
Thus the reticulation $s_{i, j_i}$ does not share a block with any reticulation contained in a variable gadget.
Furthermore, since gadgets of the same type are not connected to each other, no two reticulations within clause gadgets, and no two within variable gadgets, share a block. 
Hence $G$ is level-$1$. This completes the proof.
\end{proof}

For illustration, consider the 3-CNF formula $\varphi = (v_1 \vee v_2 \vee v_3)\wedge(\neg v_1 \vee v_3 \vee v_4)$ and the network $N(\varphi)$ shown in Figure~\ref{fig:whole network}. 
Figure~\ref{fig:whole network colored} shows the minimal support network $G$ of $N(\varphi)$ constructed according to the procedure described in the proof of Theorem~\ref{lem:satisfiable}, using the satisfying assignment $(v_1^\ast, v_2^\ast, v_3^\ast, v_4^\ast) = (\mathsf{T}, \mathsf{T}, \mathsf{T}, \mathsf{F})$ and the indices $(j_1, j_2) = (1, 2)$. 
As Figure~\ref{fig:whole network colored} shows, each reticulation of~$G$ is contained in a distinct non-trivial block (depicted in red), confirming that $G$ is indeed level-1.

\begin{theorem}
\label{lem:not satisfiable}
If a 3-CNF formula $\varphi$ is not satisfiable, then $\mathrm{level}^\ast (N(\varphi)) \geq 2$ holds.
\end{theorem}

\begin{proof}
Assume that $\varphi$ is not satisfiable.
By Lemma~\ref{lem:BN search correctness}, it suffices to show that every element in $\mathcal{B}_{N(\varphi)}$ contains a block with at least two reticulations.

Let $G$ be an arbitrary element of $\mathcal{B}_{N(\varphi)}$. 
By Lemma~\ref{lem:construction}, $G$ can be decomposed into the tree $T$, a minimal support network of each variable gadget $N'(v_m)$, and a minimal support network of each clause gadget $N(C_i)$.
By the first statement of Lemma~\ref{obs:variable gadget}, for each $m \in [1, n]$, exactly one of $G'_\mathsf{T}(v_m)$ and $G'_\mathsf{F}(v_m)$ is a subgraph of $G$.
This allows us to define a truth assignment $\mathbf{v} \colon \{v_1, \dots, v_n\} \to \{\mathsf{T}, \mathsf{F}\}$ by $\mathbf{v}(v_m) = \mathsf{T}$ if $G \supseteq G'_\mathsf{T}(v_m)$, and $\mathbf{v}(v_m) = \mathsf{F}$ otherwise.
Since $\varphi$ is not satisfiable, there exists a clause $C_i$ such that all three literals $\lambda_{i,1}, \lambda_{i,2}$ and $\lambda_{i,3}$ evaluate to $\mathsf{F}$ under~$\mathbf{v}$.
Consider such $C_i$. 

Let $j\in [1, 3]$ and suppose that $\lambda_{i,j}$ is a positive literal $v_m$. Since $v_m^\ast = \mathsf{F}$, $G'_\mathsf{F}(v_m)$ is used in $G$. Then, by our construction of $N(\varphi)$, the vertices $d_{i,j}$ and $e_{i,j}$ are located on the directed path from $a_{\neg v_m}$ to $b_{\neg v_m}$. 
This path is contained in the non-trivial block of $G'_\mathsf{F}(v_m)$ by Lemma~\ref{obs:variable gadget}.
Similarly, if $\lambda_{i,j} = \neg v_m$, then $v_m^\ast = \mathsf{T}$ and hence $G'_\mathsf{T}(v_m)$ is used in $G$; 
$d_{i,j}$ and $e_{i,j}$ are located on the directed path from $a_{v_m}$ to $b_{v_m}$. 
Lemma~\ref{obs:variable gadget} implies that this path is contained in the non-trivial block of $G'_\mathsf{T}(v_m)$.
Thus, each of the roots of $N(C_i)$ is contained in the non-trivial block of its corresponding variable gadget in $G$.

However, Lemma~\ref{obs:clause all} implies that in any minimal support network of $N(C_i)$, some reticulation $s_{i,j}$ is a descendant of  two of these six roots.
Consequently, $s_{i,j}$ belongs to the same non-trivial block as the reticulation of the corresponding variable gadget, which yields a block with at least two reticulations. This completes the proof.
\end{proof}

By Proposition~\ref{prop:Nphi is not TB}, Theorems~\ref{lem:satisfiable} and~\ref{lem:not satisfiable}, a 3-CNF formula
$\varphi$ is satisfiable if and only if the network $N(\varphi)$ is level-$1$-based.
Moreover, the network $N(\varphi)$ can be constructed in $O(n+\ell)$ time.
We also note that \textsc{Level Decision} belongs to NP for any $k\geq 0$, since the level of any
support network~$G= (V, E)$ can be computed in $O(|V|+|E|)$ time by decomposing it into blocks via depth-first search~\cite{HopcroftTarjan-1973-Algorithm447Efficient} and counting the reticulations in each block. 
Therefore, we obtain the following result.
\begin{theorem}
  \label{thm:level-1 np-h}
  \textsc{Level Decision} is NP-complete for $k=1$. Hence, \textsc{Level Minimization} is NP-hard.
\end{theorem}

\section{An Extension to the NP-completeness Proof of Level Decision for General $k$}\label{sec:extension}
In this section, we show that the proposed reduction in Section~\ref{sec:np-hardness} can be extended to establish the NP-completeness of Problem~\ref{prob:level-k-based decision} for general~$k\geq 2$. This is achieved by applying a systematic modification to each variable gadget used in our construction.

For $k \geq 2$, we define a rooted almost-binary phylogenetic network $N_k(v_m)$ on $X_{v_m} = \{x_{m,1}, x_{m,2}\}$ as follows.
We start with $N(v_m)$, which is depicted in Figure~\ref{fig:level-k variable gadget}(a) with two additional labels $c_{v_m}$ and $d_{v_m}$.
We then repeat the following operation $k-1$ times to obtain $N_k(v_m)$: replace the unique incoming edge at $d_{v_m}$ with a copy of the rooted almost-binary phylogenetic network $H$ shown in Figure~\ref{fig:level-k variable gadget}(b).
For any 3-CNF formula $\varphi$, let $N_k(\varphi)$ be the network constructed by connecting the tree $T$, the variable gadgets $N_k(v_m)$, and the clause gadgets $N(C_i)$ in exactly the same manner as in the construction of $N(\varphi)$.

\begin{figure}[h]
    \centering
    \includegraphics[scale=.38]{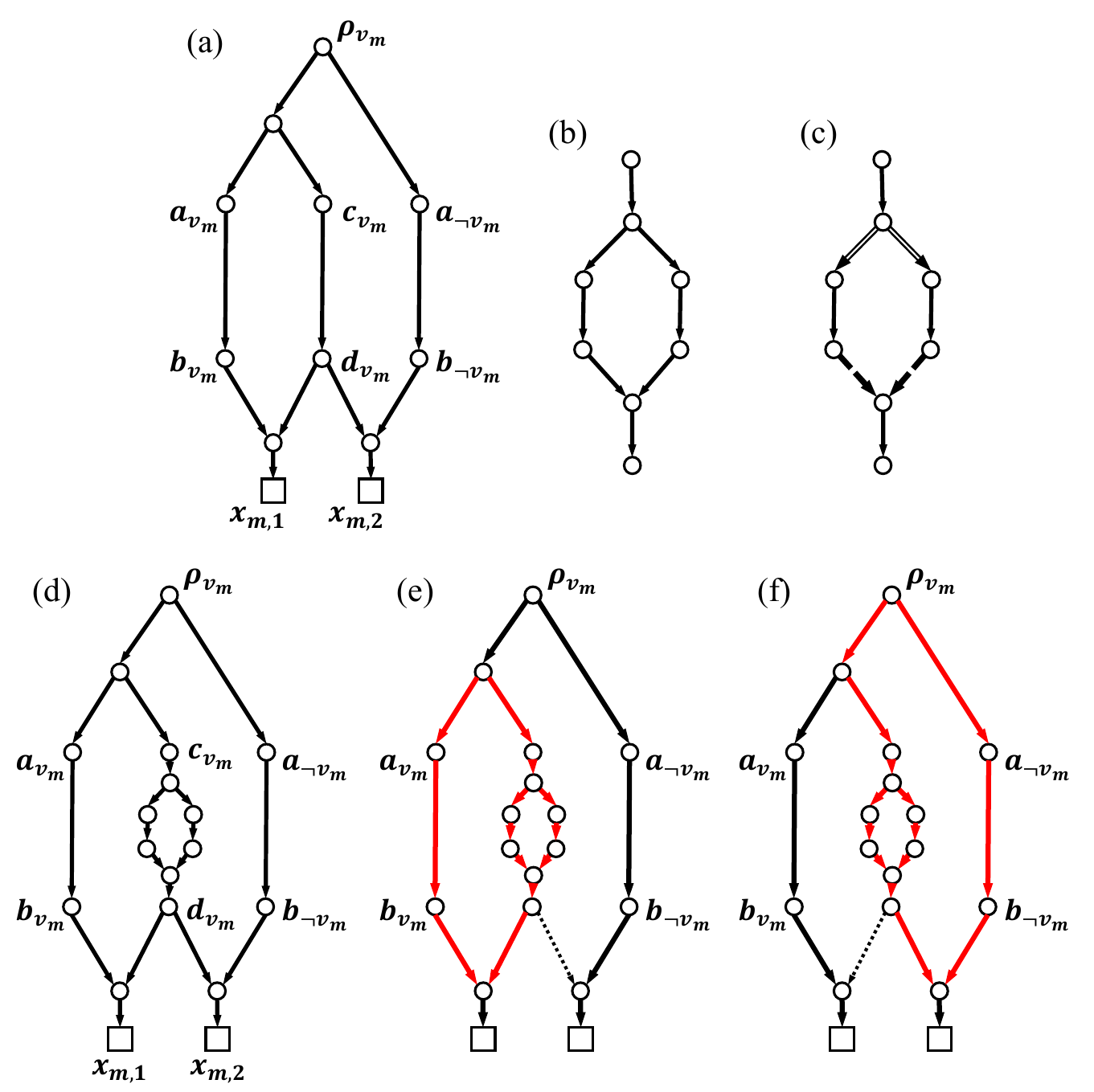}
\caption{(a)~The original variable gadget~$N(v_m)$ shown in Figure~\ref{fig:variable gadget}(a).
(b)~The network~$H$ used in our construction of $N_k(v_m)$. 
(c)~The maximal zig-zag trail decomposition of~$H$. 
(d)~The variable gadget~$N_2(v_m)$.
(e)~A minimal support network~$G_\mathsf{T}^2(v_m)$ of~$N_2(v_m)$.
(f)~A minimal support network~$G_\mathsf{F}^2(v_m)$ of~$N_2(v_m)$.
The depiction of fences and the coloring of edges follow Figure~\ref{fig:variable gadget}.}
    \label{fig:level-k variable gadget}
\end{figure}

Each variable gadget $N_k(v_m)$ contains $k-1$ copies of $H$.
Figure~\ref{fig:level-k variable gadget}(c) shows that each copy of $H$ consists of one M-fence of length~$2$, one W-fence of length~$2$, and four N-fences of length~$1$. 
All of its edges are then contained in every minimal support network of $N_k(v_m)$.
Consequently, $N_k(v_m)$ has exactly two minimal support networks, denoted by $G_\mathsf{T}^k(v_m)$ and $G_\mathsf{F}^k(v_m)$, which are obtained by applying the same attachment to $G_\mathsf{T}(v_m)$ and $G_\mathsf{F}(v_m)$, respectively.
We illustrate $N_k(v_m), G_\mathsf{T}^k(v_m)$ and $G_\mathsf{F}^k(v_m)$ for $k=2$ in
Figure~\ref{fig:level-k variable gadget}(d), (e), and (f), respectively.

We now examine how the arguments of Lemma~\ref{lem:construction}, Proposition~\ref{prop:Nphi is not TB}, Theorems~\ref{lem:satisfiable} and~\ref{lem:not satisfiable} change when $N(\varphi)$ is replaced by $N_k(\varphi)$, and accordingly each $N(v_m)$, $G_\mathsf{T}(v_m)$, and $G_\mathsf{F}(v_m)$ is replaced by $N_k(v_m)$, $G_\mathsf{T}^k(v_m)$, and $G_\mathsf{F}^k(v_m)$, respectively. First, since the gadgets are connected to one another at the same locations as in the construction of $N(\varphi)$, which differ from the locations where the copies of $H$ are inserted in each $N_k(v_m)$, the argument of Lemma~\ref{lem:construction} carries over to $N_k(\varphi)$ without modification. 
The analogue of Proposition~\ref{prop:Nphi is not TB} reads as follows: for each $m \in [1, n]$, both $G_\mathsf{T}^k(v_m)$ and $G_\mathsf{F}^k(v_m)$ contain a block of level $k$, and hence $\mathrm{level}^\ast(N_k(\varphi)) \geq k$.
If $\varphi$ is satisfiable, then the level-$1$ support network $G$ of $N(\varphi)$ constructed in the proof of Theorem~\ref{lem:satisfiable} becomes a level-$k$ support network of $N_k(\varphi)$ when constructed analogously, since each variable gadget now contributes a block of level $k$ instead of $1$. If $\varphi$ is not satisfiable, then the argument in the proof of Theorem~\ref{lem:not satisfiable} extends to $N_k(\varphi)$ to show that every minimal support network of $N_k(\varphi)$ contains a block with at least $k+1$ reticulations.
Combining these observations yields the following corollary.

\begin{corollary}
  \label{cor:level-k-based iff}
If a 3-CNF formula $\varphi$ is satisfiable, then $\mathrm{level}^\ast (N_k(\varphi)) = k$ holds.
Otherwise, $\mathrm{level}^\ast(N_k(\varphi)) \geq k+1$ holds.
\end{corollary}

Similarly to the construction of $N(\varphi)$, one can construct $N_k(\varphi)$ in $O(nk+\ell)$ time. Hence we obtain the following theorem.
\begin{theorem}
  \label{thm:level-k-based hardness}
  \textsc{Level Decision} is NP-complete for every fixed $k\geq 2$.
\end{theorem}

\section{Conclusion}\label{sec:conclusion}
In this paper, we have settled the conjecture posed in~\cite{SuzukiEtAl-2025-WhichPhylogeneticNetworks} by proving that \textsc{Level Minimization} is NP-hard.
In fact, we have established the stronger result that \textsc{Level Decision} is NP-complete for every fixed integer $k \geq 1$ (Theorem~\ref{thm:level-1 np-h} and~\ref{thm:level-k-based hardness}).
As an immediate consequence, the exact exponential-time algorithm for \textsc{Level Minimization} developed in~\cite{SuzukiEtAl-2025-WhichPhylogeneticNetworks} cannot be improved to a polynomial-time algorithm, unless $\mathrm{P} = \mathrm{NP}$.
Moreover, our hardness result contrasts with the fact that a closely related problem, finding a support network with the fewest reticulations, can be solved in linear time~\cite{SuzukiEtAl-2025-WhichPhylogeneticNetworks}.

A natural direction for future work is to investigate the parameterized complexity of \textsc{Level Minimization}.
Our hardness result rules out fixed-parameter tractability when parameterizing by the base level, unless $\mathrm{P} =\mathrm{NP}$, so future work should instead consider parameters of the input network.
Promising candidates are the \emph{treewidth} and its DAG-oriented variant, the \emph{scanwidth}, both of which tend to be small on phylogenetic networks arising from real biological data~\cite{berry2020scanning}.

Another direction is to seek combinatorial characterizations, or at least necessary conditions, of level-$k$-based networks. Tree-based networks, which coincide with level-0-based networks, admit such a characterization through the maximal zig-zag trail decomposition, namely via the absence of W-fences (Proposition~\ref{prop:tree-based}). It would be of interest to obtain analogous results for level-1-based networks and, more generally, for level-$k$-based networks for each $k\geq 1$. We note that our NP-hardness result rules out characterizations as simple as the one for the tree-based networks.
However, even necessary conditions could lead to faster algorithms for \textsc{Level Decision} by further restricting the search space beyond $\mathcal{B}_N$.

\section*{Acknowledgments}
The author is grateful 
to Momoko Hayamizu, Jurre van Schie, Mike Steel, and participants of the Combinatorial Mathematics Seminar (COMA SEMI) for helpful discussions.
This work is a part of the outcome of research performed under a Waseda University Grant for Special Research Projects (Project number: 2026C-088).

\bibliographystyle{unsrt}
\bibliography{takatora-counting}

\end{document}